\documentclass[aps,pra,superscriptaddress,nofootinbib,twocolumn]{revtex4-2}

\usepackage{qcircuit}
\usepackage[dvips]{graphicx}
\usepackage{amsmath,amssymb,amsthm,mathrsfs,amsfonts,dsfont}
\usepackage{subfigure, epsfig}
\usepackage{braket}
\usepackage[colorlinks=true,citecolor=
blue,linkcolor=magenta]{hyperref}
\usepackage{bm}
\usepackage{enumerate}
\usepackage{color}
\usepackage{physics}
\usepackage{comment}

\newtheorem{theorem}[]{Theorem}
\newtheorem{definition}[]{Definition}
\newtheorem{lemma}[]{Lemma}

\begin{document}

\title{Quantum computing is scalable on a planar array of qubits with fabrication defects}

\author{Armands Strikis}
\email{armands.strikis@mansfield.ox.ac.uk}
\affiliation{Department of Materials, University of Oxford, Oxford OX1 3PH, United Kingdom}

\author{Simon C. Benjamin}
\affiliation{Department of Materials, University of Oxford, Oxford OX1 3PH, United Kingdom}
\affiliation{Quantum Motion, Pearl House, 5 Market Road, London N7 9PL, United Kingdom}

\author{Benjamin J. Brown}
\affiliation{Centre for Engineered Quantum Systems, School of Physics, University of Sydney, New South Wales 2006, Australia}
\affiliation{Niels Bohr International Academy, Niels Bohr Institute, Blegdamsvej 17, University of Copenhagen, 2100 Copenhagen, Denmark}

\begin{abstract}
To successfully execute large-scale algorithms, a quantum computer will need to perform its elementary operations near perfectly. This is a fundamental challenge since all physical qubits suffer a considerable level of noise. Moreover, real systems are likely to have a finite {\it yield}, i.e. some non-zero proportion of the components in a complex device may be irredeemably broken at the fabrication stage. We present a threshold theorem showing that an arbitrarily large quantum computation can be completed with a vanishing probability of failure using a two-dimensional array of noisy qubits with a finite density of fabrication defects. To complete our proof we introduce a robust protocol to measure high-weight stabilizers to compensate for large regions of inactive qubits. We obtain our result using a surface code architecture. Our approach is therefore readily compatible with ongoing experimental efforts to build a large-scale quantum computer.
\end{abstract}

\maketitle

\section{Introduction}

Full-scale quantum information processing will require a large number of physical qubits, to facilitate error correction protocols that will overcome errors experienced by the machine~\cite{Shor1996}. A approach adopted by multiple research teams across various platforms is to fabricate a two-dimensional device with a large array of interacting qubits realised on its surface~\cite{Huang2020, Zhang2018, Bruzewicz2019}. However, as we aim to produce larger devices with more physical qubits, we are likely to find that some non-zero fraction of the components we rely upon to realise qubits will be highly-imperfect, or will not function at all. We must therefore find error-correction protocols that enable fault-tolerant quantum computing using these flawed devices. Such a method might even provide solutions to two  closely aligned problems: First, there is the challenge that even when all components function nominally at the start of a calculation, there may be  mid-calculation events that lead to long-lasting (if not permanent) defects in the system. Examples include cosmic ray impacts within superconducting or silicon devices, or rare leakage or loss events in ion trap or neutral atom arrays~\cite{Martinis2021, Wilen2021, McEwen2021, Vala2005, Bermudez2017,rydberg2021}. Second, in a given platform it may be problematic to realise an unbroken homogeneous two-dimensional plane of interlinked qubits over very large scales; instead it may be desirable to introduce void spaces through which control or power line bundles, cooling channels or other infrastructure can pass. In effect, one would be deliberately engineering in rare but large `defects' in the sense explored in this paper.

Quantum error correction can proceed by repeatedly performing stabilizer measurements, specified by some code, that identify the errors occurring over time~\cite{Gottesman1997, Kitaev1997, Kitaev03, Dennis2002, wang2003confinement, Fowler2012, Lidar2013}. These measurements produce a syndrome pattern that we use to determine how to prevent errors at the logical level. Provided the physical errors act locally and at a sufficiently low rate, we can make the failure rate of a quantum computation arbitrarily small by increasing the size of the codes we use. However, introducing a fixed architecture with a finite density of fabrication defects would compromise our ability to measure all of the stabilizers of the code. We will then require suitably-adapted strategies to collect enough syndrome information to recover the output of a quantum computation reliably.

In this work we show how to perform quantum computation with an arbitrarily small logical error rate using a two-dimensional array of qubits and a finite density of fabrication defects. We specifically consider error correction with the surface code~\cite{Kitaev03, Fowler2012}, a code that is now under intensive experimental development~\cite{Takita2017, Chen2021, Andersen2020}. 
Effectively, fabrication defects introduce punctures to the qubit array, some of which may be very large, where we assume we have no control over qubits within each puncture. An essential component to obtain our result is the design of a protocol to collect syndrome data reliably near to where fabrication defects lie. Given this component, we show that we can successfully complete a computation running on the remaining intact qubits on the lattice  with high probability, assuming these qubits experience errors at a suitably low rate. Our result is proved assuming a phenomenological noise model capturing the key features of a noisy quantum circuit where qubits experience local errors at a low rate, and where measurements can give incorrect outcomes if errors occur on the circuit elements used in their readout circuits.

Earlier work has shown that single-shot quantum error-correcting codes are intrinsically robust to fabrication defects, or more generally, time-correlated errors~\cite{Bombin2016}. However, known single-shot codes are realised in no fewer than three dimensions. As such, their manufacture will present technological challenges compared with their two-dimensional counterparts. For small two-dimensional systems suffering fabrication defects, numerical results have shown that the logical failure rate can be suppressed by increasing the code distance~\cite{Stace2009,Auger2017,Nagayama2017}. However it remains to be shown whether these solutions scale to give an arbitrarily low logical failure rate in the presence of a finite density of fabrication defects. Indeed, it is commented that some of these results may only demonstrate a pseudo-threshold~\cite{Auger2017}. This observation is consistent with recent work that argues that a two-dimensional code with a high density of static punctures, arranged in a fractal pattern, will not demonstrate a finite threshold error rate~\cite{zhu2021t}. In the present paper we use analytical methods together with a dynamic protocol to reliably learn the values of high-weight stabilizers. We thus build upon earlier proposals~\cite{Stace2009,Auger2017,Nagayama2017} to obtain, for the first time, a threshold theorem for two-dimensional systems suffering a finite rate of fabrication defects.

The paper is structured as follows. In Section~\ref{sec:UCnCD} we describe error correction with fabrication defects and introduce a protocol to measure syndrome data near to fabrication defects. In  Section~\ref{sec:TT} we introduce the technical tools we need to prove a threshold theorem. In Section~\ref{sec:FTQC} we argue that our results are general to arbitrary quantum computations with the surface code model. We discuss the applicability of our protocol to cosmic ray and other time-correlated errors in Section~\ref{sec:CR}. In Section~\ref{sec:Disc} we conclude with a discussion on future work.

\section{Error correction with fabrication defects}
\label{sec:UCnCD}

We begin by giving a brief overview of error correction with the surface code before explaining how we adapt the system to compensate for fabrication defects. We advise the reader looking for more details about the review material on surface-code error correction to see Refs.~\cite{Kitaev03, Dennis2002, wang2003confinement}.
We define the surface code with a qubit on each of the edges of a square lattice. It is a stabilizer code with two types of stabilizer generators; star operators $A_v = \prod_{\partial e \ni v} X_e$ at vertices $v$ and plaquette operators $B_w = \prod_{e \in \partial w} Z_e$ at faces $w$ where $ X_e$ and $Z_e$ denote the standard Pauli matrices acting on qubit $e$, $\partial e$ denote the set of vertices $v$ on the boundary of edge $e$ and $ \partial w$ denote the set of edges $e$ on the boundary of face $w$. We can also define two types of boundary, rough and smooth, that, respectively, have modified weight-three Pauli-Z and Pauli-X stabilizers where the lattice terminates~\cite{bravyi1998, Dennis2002}.
The stabilizers specify a code subspace spanned by state vectors $|\psi\rangle$ satisfying $A_v|\psi\rangle = B_w |\psi\rangle = (+1) |\psi \rangle$. 
The distance of the code $d$ is the weight of the least weight non-trivial, i.e. non-identity, Pauli logical operator where logical operators commute with all of the stabilizer operators, but are not themselves members of the abelian group generated by the stabilizers $A_v,B_w \in \mathcal{S}$. The weight of a Pauli operator is the number of qubits in its non-trivial support.

Fault-tolerant error correction with the surface code \cite{Dennis2002, wang2003confinement, raussendorf05, Fowler2012} is described by a $(2+1)$-dimensional space-time lattice where the time axis runs in a direction orthogonal to the two-dimensional planar qubit array that supports the surface code. Stabilizers generators of the qubit array are measured repeatedly over time to identify errors.
We say that we detect an error event when the outcome of a given stabilizer differs over two consecutive measurements; thus the location of the detection event, or just `event' for short, is defined both in space and time. 
We call the configuration of all the detection events on the space-time lattice the error syndrome.

\begin{figure}
    \centering
    \includegraphics{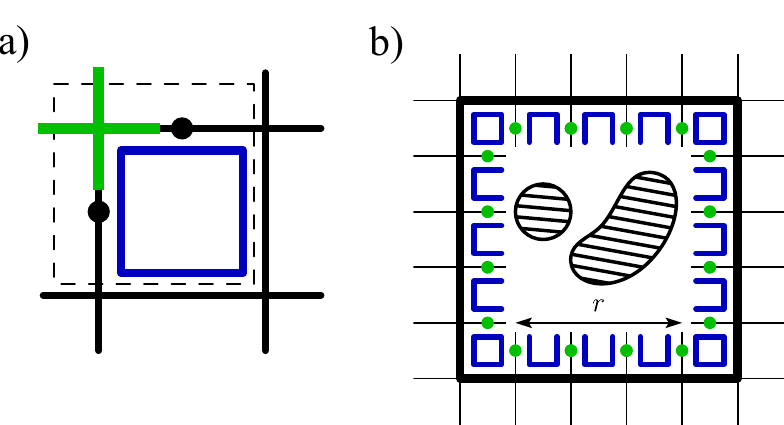}
    \caption{Dealing with fabrication defects. (a)~A basic unit cell of the lattice that consists of two qubits together with a single star operator, shown in green, and a plaquette operator, blue, enclosed in the dashed line. (b)~We measure two high-weight super stabilizers $\mathcal{A}_P$ and $\mathcal{B}_P$ to detect errors near to the puncture $P$ associated to the cluster of fabrication defects. Boundary qubits $e \in \partial P$ are marked by green spots and boundary plaquettes $w \in \partial P$ are shown in blue. The super stabilizer $\mathcal{A}_P$ can be inferred from the product of Pauli-X measurements $e \in \partial P$, and the other, $\mathcal{B}_P$, is the product of Pauli-Z operators on the qubits marked with a thick black line. We can infer the value of $\mathcal{B}_P$ with the product of plaquettes $w\in \partial P$ by the relationship $\mathcal{B}_P = \prod_{w \in \partial P} B_w$. }
    \label{Fig:UnitCell}
\end{figure}

Errors on the surface code can be regarded as strings in the space-time lattice that produce detection events at their endpoints. 
Bit-flip and phase-flip errors that occur on the qubit array over time produce strings that run parallel to the surface code in space time. Measurement errors, i.e. errors that cause a stabilizer measurement to return the incorrect outcome, can be regarded as 
string segments that run parallel with the temporal axis. In general, individual errors compound to make longer strings. We correct errors using the space-time picture by collecting a large history of syndrome data. We look for a correction that matches together pairs of nearby detection events, such that the error, together with its correction, acts trivially with the encoded information. This is done with a classical software called a decoder. The minimum-weight perfect matching (MWPM) decoder~\cite{Dennis2002, wang2003confinement, raussendorf05, Fowler2012} is one such tool that operates by matching pairs of detection events with error strings in a way that minimises the total weight of these strings. We consider this decoder in our proof for a threshold. Importantly, it has the capability of finding a correction that corresponds to the least-weight error that may have caused a given syndrome.

\subsection{Fabrication defects and super stabilizers}

We introduce shells in the space-time lattice to detect error events near to fabrication defects. Before we describe shells, here, we first explain how fabrication defects manifest themselves in the qubit array. We also describe our protocol to measure high-weight super stabilizers that enclose the punctures created by fabrication defects. In the following Subsection we explain how our robust method for measuring super-stabilizers in the qubit array will give rise to shells in the space-time lattice.

We coarse grain the qubit array in terms of unit cells, see Fig.~\ref{Fig:UnitCell}(left). The unit cell consists of four elements; two data qubits, one star operator, and one plaquette operator. This enables us to define our model for fabrication defects: We say that a unit cell supports a fabrication defect unless all of its components are intact. That is, its qubits, and the circuitry and ancilla qubits used to measure stabilizers all function correctly. However if even one of the elements is not intact then, for the purposes of the present paper, we assume that the entire set of four elements cannot be used. We also assume that we can use the circuitry of intact unit cells to measure stabilizers with a reduced support near to a fabrication defect. Specifically, we will make weight-three measurements near to fabrication defects, where these measurements are contained in the support of weight-four stabilizer measurements of an ideal qubit array.

Effectively, fabrication defects give rise to punctures on the lattice where the standard local stabilizer generators of the ideal code cannot be measured.
Nevertheless, it is still important to identify events at the endpoints of string-like errors at the locations of these punctures. We therefore need to use alternative strategies to determine the locations of events on areas of the lattice where fabrication defects exist. Our protocol makes use of super-stabilizers that enclose fabrication defects, as proposed in Refs.~\cite{Stace2009,Auger2017, Nagayama2017}. Super stabilizers are supported on the boundary of fabrication defects, and identify events inside the puncture.

\begin{figure}
    \centering
    \includegraphics{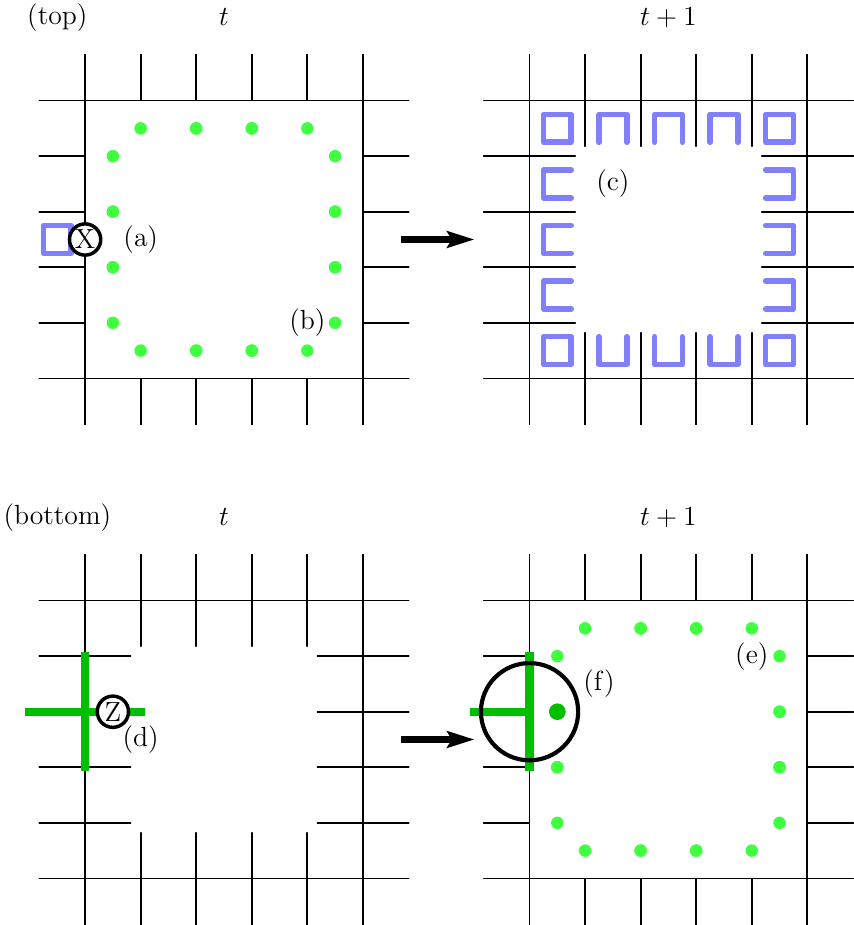}
    \caption{Measuring super stabilizers that enclose fabrication defects. (top)~A code deformation from a smooth boundary to a rough boundary. A Pauli-X error is detected by a plaquette, (a). The code deformation is carried out by initialising the green qubits in the $|+\rangle$ state, (b). The Pauli-X error is detected by the super stabilizer after the code deformation where we measure the boundary plaquette operators,~(c). (bottom)~A code deformation from a rough boundary to a smooth boundary. (d)~A Pauli-Z error at the rough boundary is detected by a weight-four star operator. (e)~The Pauli-Z error is also detected by the super stabilizer. (f)~A Pauli-Z error that occurs just before the deformation is identified by the weight-four check that is inferred from a single-qubit measurement and a weight-three stabilizer measurement of the smooth boundary.}
    \label{Fig:Deformations}
\end{figure}

To correct the surface code reliably, we need to measure the values of two high-weight stabilizer operators that enclose each puncture $P$ on the lattice; one star operator and one plaquette operator. We construct a square boundary of linear size $r$ that completely encloses the fabrication defects, see fig.~\ref{Fig:UnitCell}(right). We denote the qubits and stabilizers associated to the boundary with the set $ \partial P $. We show the boundary qubits of $\partial P$ with green spots and its boundary plaquettes in blue in Fig.~\ref{Fig:UnitCell}(right). We minimise $r$ such that all of the qubits and plaquettes on the boundary $\partial P$ are intact and assume the small cluster of fabrication defects that give rise to $P$ is well separated from all other fabrication defects. The two boundary stabilizers about puncture $P$ are $\mathcal{A}_P$; the product of Pauli-X operators on the qubits marked by green spots and $\mathcal{B}_P$ which is the product of Pauli-Z terms supported on the thick black line in Fig.~\ref{Fig:UnitCell}(right).

We adopt the method used in Refs.~\cite{Stace2009,Auger2017,Nagayama2017} to measure super stabilizers. It is shown that these complicated measurements can be decomposed into several low-weight non-commuting `gauge' measurements to read out these stabilizers in a practical way. We infer the value of $\mathcal{A}_P = \prod_{e \in \partial P} X_e$ by measuring all of the qubits $e \in \partial P$ in the Pauli-X basis. Similarly, we can infer the value of $\mathcal{B}_P$ with low-weight gauge operators. We have that $\mathcal{B}_P = \prod_{w \in \partial P} B_w$ where now $\partial P$ is the set of plaquettes at the boundary of the puncture shown in blue in Fig.~\ref{Fig:UnitCell}(right). We can therefore infer the value of $\mathcal{B}_P$ by measuring the plaquette operators in $\partial P$.

We must regularly measure both of these high-weight stabilizer operators accurately to identify errors that occur over time. However, we cannot learn the values of all of the gauge checks simultaneously since they do not commute. Furthermore, in practice, error-detecting parity measurements may be unreliable and return incorrect outcomes. To deal with these issues we divide the gauge measurements into two commuting subsets associated to each of the two different stabilizer measurements about each puncture. The first of these subsets are Pauli-X measurements $X_e$ for $e\in \partial P$ and the other are plaquette measurements $B_w$ for all $w\in\partial P$. Note all of these measurements commute with the standard stabilizers on the lattice and therefore either subset can be measured simultaneously with the other stabilizers of the code. Likewise, super-stabilizers $\mathcal{A}_P$ and $\mathcal{B}_P$ both commute with all the gauge measurements.

Measuring one subset of gauge measurements followed by the other, and so on, acts like a code deformation of the stabilizer group near to the puncture where we transform a rough boundary into a smooth boundary, and vice versa. Let us look at these deformations in detail, and explain how different types of errors are identified as we deform the code. We find that all types of individual errors create detection events in pairs, thereby allowing us to employ standard minimum-weight perfect matching to decode the error syndrome in space time.

Let us begin by looking at how we detect an error with the super stabilizer $\mathcal{B}_P$, see Fig.~\ref{Fig:Deformations}(top). The code deformation changes the boundary around the puncture from smooth to rough, whereby after the deformation we are in a known eigenstate of all the stabilizers corresponding to the lattice shown at time $t+1$. At time $t$ we measure the stabilizers of the surface code where the puncture has a smooth boundary, and at the same time we initialise all of the qubits $e\in\partial P$, shown by green spots in Fig.~\ref{Fig:Deformations}(b), in the $|+\rangle$ state. This choice for initialisation means that we begin in a known eigenstate of the weight-four star operators at the rough boundary shown at time $t+1$ as well as the operator $\mathcal{A}_P$.

We assume a Pauli-X error occurs at the smooth boundary at time $t$ or earlier, see Fig.~\ref{Fig:Deformations}(a). This is detected by the weight-four plaquette stabilizer, also shown at Fig.~\ref{Fig:Deformations}(a), where the event is detected at the time the error occurs. 
We also obtain a second detection event inside the puncture that we measure during the code deformation. At time $t+1$ all of the boundary plaquettes, $B_w $ for $w\in\partial P$ are measured, thereby projecting the lattice into a known eigenstate of the weight-three plaquette operators, shown in Fig.~\ref{Fig:Deformations}(c), assuming for now that all of the measurements are reliable. This deformation also reveals the value of $\mathcal{B}_P$ using that $\mathcal{B}_P = \prod_{w \in \partial P} B_w$. We therefore identify a second event that can be regarded as lying inside the puncture since $\mathcal{B}_P$ anti-commutes with the Pauli-X error.

We will also detect an event when we measure $\mathcal{B}_P$ if a $B_w $ operator for some $w\in\partial P$ experiences a measurement error. We identify a second event in the space-time lattice for one such measurement error by repeating the stabilizer measurements of the code with the rough boundary puncture we have just produced. We compare the measurement results for the $B_w$ operators at time $t+2$ with the measurement outcomes collected at time $t+1$ to identify a measurement error that occurred at time $t+1$ where we projected the boundary of the puncture onto a rough boundary. Unlike when the code deformation is initially performed, after measuring the $B_w$ operators at time $t+1$ the system is projected into an eigenstate of all $B_w$ operators at the boundary $\partial P$. Upon a second round of measurements then, we can identify events produced by individual measurement errors that occurred for measurements performed at time $t+1$ by comparing the measurements at the two concurrent times.

In Fig.~\ref{Fig:Deformations}(bottom) we deform a rough boundary enclosing the fabrication error, shown at time $t$, onto a smooth boundary. To make this deformation we measure all of the qubits $e \in \partial P$ in the Pauli-X basis at time $t+1$. We consider a Pauli-Z error on a qubit $e\in \partial P$ before the deformation, see Fig.~\ref{Fig:Deformations}(d). This error creates a detection event at the weight-four star operator $A_v$ shown at Fig.~\ref{Fig:Deformations}(d). A second event is also created in the puncture that is detected when we measure $\mathcal{A}_P$. 
$\mathcal{A}_P$ is given by the product of single qubit measurements during the deformation. We note that at time $t+1$ we also begin measuring the weight-three star operators associated to the smooth boundary at the puncture we have just created. We infer the value of this weight-three operator using the value of the single-qubit measurement and the weight-four star operator before the code deformation.

Let us finally consider the situation where an error occurs immediately before the code deformation takes place, where a Pauli-Z error occurs after time $t$ but before time $t+1$ when the deformation is performed. This error is detected by the inferred measurement the weight-four star operator at the rough boundary. However, once the code deformation takes place, we identify the bulk detection event by inferring the value of the weight-four star operator. We infer its value by taking the product of the weight-one measurement and the weight-three operator of the smooth boundary shown in the circle in Fig.~\ref{Fig:Deformations}(f). This measurement will reveal a second event that can be paired to the detection event identified by the super stabilizer $\mathcal{A}_P$.

One should also consider the measurement errors that occur when we perform this deformation. A measurement error on the weight-three gauge measurement will create a detection event at this space-time location of the stabilizer shown in Fig.~\ref{Fig:Deformations}(d). In the case of this measurement error on the weight-three-check, a second event will be identified when the weight-three measurement is repeated, assuming we do not perform any gauge measurements that do not commute with the weight-three operator in the interim.
A measurement error on one of the single-qubit Pauli-X measurements is indistinguishable from the Pauli-Z error on a boundary qubit $e\in \partial P$ that occurs immediately before the deformation is performed.

\begin{figure*}[tbp]
    \centering
    \includegraphics[width =1\linewidth]{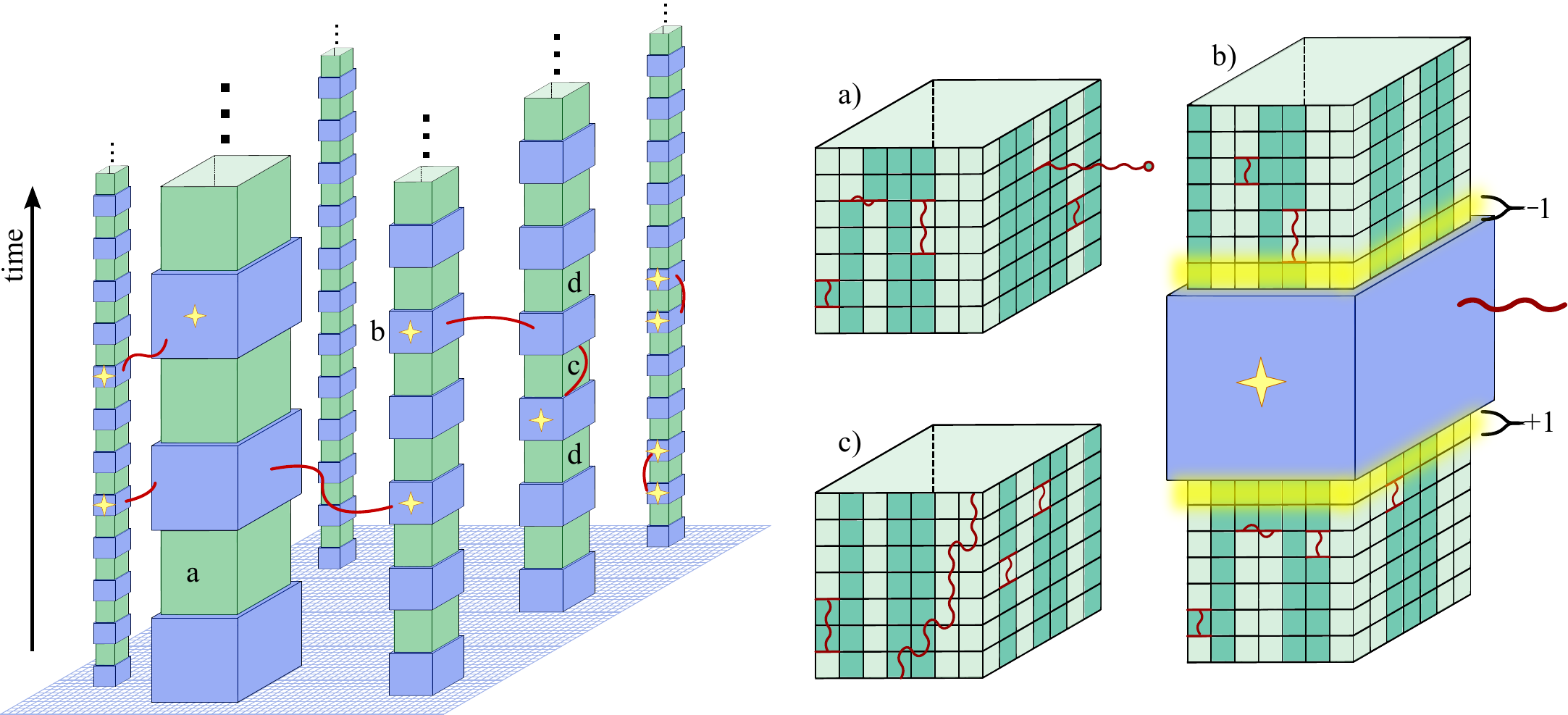}
    \caption{Space-time lattice with super-stabilizers forming shells around the fabrication defect clusters.
    We measure two different super stabilizers $\mathcal{A}_P$ and $\mathcal{B}_P$ in the space-time lattice about each puncture $P$, shown by columns that align along the temporal direction.
    We alternate between measurements of $\mathcal{A}_P$ and $\mathcal{B}_P$ over time for each puncture, where we depict the measurement of either $\mathcal{A}_P$ or $\mathcal{B}_P$ when each column is coloured either green or blue, respectively.
    The choice of super-stabilizer alternates less frequently for larger clusters of fabrication defects, so as to deal with the increased number of errors that can occur while measuring a higher-weight super stabilizer. The red strings show a specific error configuration, with stars indicating the shells that detect an event at the terminal point of an error string. The right panels describe common circumstances; (a) a typical shell with few gauge-operator measurement errors and a bit-flip error string (with light green and dark green squares representing $+1$ and $-1$ gauge-operator outcomes respectively), (b) the value of super-stabilizers before and after a shell determine that a bit-flip error string has terminated at the smooth shell between them and (c) a collection of errors that results in an error string across the shell; these errors are detected by comparing the values of the previous and next shell of the same boundary ($d$ in the left panel).}
    \label{fig:3dVis}
\end{figure*}

\subsection{Shells}

We have described two operations to deform a boundary around a puncture from rough to smooth and vice versa. Furthermore, we have discussed how measurement errors can affect the outcome of the measurement of super stabilizers, and we require a strategy to learn their values reliably. We find that we can demonstrate a threshold by repeating the stabilizer measurements of each boundary type over a time that scales with the size of the puncture before transforming the puncture to the other type. Specifically, given a puncture of size $r$, we measure the boundary stabilizers of a given boundary type over $O(r)$ rounds of stabilizer readout measurements before transforming the boundary into its opposite type. In this Subsection we describe the objects that are obtained in the space-time lattice with our strategy for repeating gauge measurements at the boundary of the punctures created by fabrication defects. Our strategy gives rise to objects in the space-time lattice that we call shells. We define shells in more detail below. We use the properties of shells that we outline here to present a threshold theorem in the following Section.

Other approaches for identifying events near to punctures have been tested numerically in Refs.~\cite{Stace2009,Auger2017,Nagayama2017} where, in contrast to our approach, the boundary type is changed at every round of stabilizer readout independent of the size of the puncture. It remains to be shown if these strategies demonstrate a threshold as the system size diverges~\cite{Auger2017}. However, in the limit where the code distance is small and we only have small fabrication defects, these numerical studies are consistent with our strategy for error-correction on a defected lattice. We also remark on complementary work~\cite{Higgott2021} where the alternative gauges of a subsystem code are fixed with an irregular frequency to concentrate low-weight error-correction measurements on detecting dominant types of errors in a biased noise model.

In the space-time picture, punctures in the qubit array are projected along columns, see Fig.~\ref{fig:3dVis}. We show the columns in alternating colours, where the green blocks indicate periods where we measure the stabilizers of a rough boundary and blue blocks indicate periods where we measure smooth boundary stabilizers. We refer to a single block as a \textit{shell}. The height of each shell is proportional to its width, showing that we measure the stabilizers of a given boundary type for a number of rounds approximately equal to its width, $\sim r$. Our code deformations are such that we can identify a single detection event at any given shell.

Strings of Pauli-X errors can terminate at smooth boundaries that are coloured blue in Fig.~\ref{fig:3dVis}. To calculate the value of a blue shell, we compare the product of the plaquette operator measurements $\partial P $ at the moment the code deformation takes place, Fig.~\ref{Fig:Deformations}(top), to the round of measurements we performed before we deformed the puncture onto a smooth boundary, as in Fig.~\ref{Fig:Deformations}(bottom); see also Fig.~\ref{fig:3dVis}(b). The product of all of these measurements should give even parity if no error strings have terminated at the shell. However, the parity of this measurement will change if any Pauli-X strings terminate at the given blue shell, or if any of the plaquette measurements experience a measurement error. In the case that we find odd parity, we mark a syndrome event at the shell that can be paired with another event.

Likewise, Pauli-Z strings can terminate on rough boundary, which correspond to green shells. The product of the single-qubit Pauli-X measurements used to deform the rough boundary onto a smooth boundary, shown in Fig.~\ref{Fig:Deformations}(bottom), give us the value of each green shell. Specifically, it detects the parity of Pauli-Z strings that have terminated at the shell, together with any measurement errors that occur on the single-qubit Pauli-X measurements, as well as any errors that occur intialising the qubits in the $|+\rangle$ state when we deform the smooth boundary onto the rough boundary, as shown in Fig.~\ref{Fig:Deformations}(middle).

We cannot necessarily measure the value of a shell when the system is initialised. For instance, at the first moment we deform a smooth boundary onto a rough boundary, we cannot compare the value of $\mathcal{B}_P$ to an earlier measurement of the same operator, because we have not initialised the system into an eigenstate of $\mathcal{B}_P$ yet. At these locations, we have to regard the shell as an extension of the boundary in the space-time lattice where the system is initialised. Similarly, if the fabrication errors are too close to one of the physical boundaries of the qubit array, we might not have enough space to fully enclose them in a shell. In this case the collection of fabrication defects are regarded as the extension of the spatial boundary by walling them off with the respective boundary. This strategy for handling fabrication errors near boundaries was employed in Ref.~\cite{Auger2017} where details can be found.

The repetition of stabilizer measurements for a given mode of boundary is a key feature of our protocol; it enables us to reliably identify the locations of all types of errors near fabrication defects. In turn we are able prove a threshold theorem. Importantly, assuming the punctures on the qubit array are well separated, then the repetition means that two shells of the same type are also well separated in the temporal direction of the space-time lattice. In Fig.~\ref{fig:3dVis} we show an error configuration that creates syndrome detection events on shells that terminate an odd parity of error strings.

\section{Threshold Theorem}
\label{sec:TT}

\begin{figure}[tbp]
    \centering
    \includegraphics[width =0.7\linewidth]{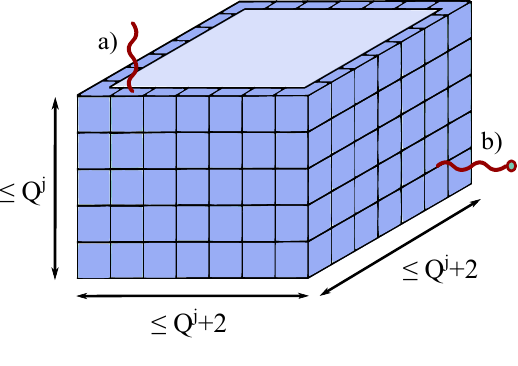}
    \caption{A smooth shell of level $j$. (a) An example temporal error that terminates at the shell due to an erroneous Z gauge measurement. (b) An example bit-flip spatial error string that terminates at the shell.}
    \label{fig:TerminationP}
\end{figure}

In this Section we develop the tools we need to prove a theorem for fault-tolerant quantum computation on a planar qubit array with a finite density of fabrication defects undergoing a phenomenological noise model. Here we show explicitly how our tools are applied to prove  the surface code undergoing an identity operation is robust. In the following section, Section~\ref{sec:FTQC}, we show how the same tools are applied to prove the theorem for a complete gate set. We begin by describing the distribution of fabrication defects and the phenomenological error model before proving our main result.

\subsection{Defect and error models}
\label{subsec:ErrorModel}

We start by describing the key features of the distribution of fabrication defects acting on the qubit array. We draw on the work of Ref.~\cite{Bravyi2011} to characterise the independent and identically distributed configuration of fabrication defects, see Proposition 7 in Appendix B of Ref.~\cite{Bravyi2011} for details. Here we briefly summarise their results. Specifically, it is shown that we can decompose a configuration of defects into connected components of unit cells, that we call \textit{clusters}. Clusters are characterised by different length scales, or \textit{levels}, that are indexed by integers $0\le j \le m $ such that a level $j$ component has a linear length of at most $Q^j$, and is separated from all other clusters of level  $k\ge j$ by distance at least $Q^{j+1} / 3$ where $Q \geq 6 $ is an integer constant that we can choose. We find that our qubit array demonstrates a threshold for quantum computing given that the array is operable in the following sense:
\begin{definition}[Operable qubit array]
We say that the qubit array is operable if its highest cluster level $m$ satisfies $Q^m \leq d / 5 -3$.
\end{definition}

Indeed, given a low enough fabrication defect rate, Ref.~\cite{Bravyi2011} demonstrates the probability of finding a cluster of linear size larger than $ d/4 - 1$ is exponentially suppressed. Specifically, given a qubit array of volume $V \sim d^2$, the probability of the occurrence of a level $m$ cluster with $Q^m > d/4 -1$ is upper bounded by $p_{fail} = \mathrm{exp}(-\Omega(d^\sigma))$ where $\sigma \approx 1/\log Q$ as long as the fabrication defect rate $f$ on any given unit cell is smaller than $f_0 = (3Q)^{-4}$. See Ref.~\cite{Bravyi2011} for proof.

As we have assumed that we can test a qubit array for the locations of fabrication defects offline, we can discard any devices that have clusters of fabrication defects larger than $d/2-1$, but as we have explained, given a sufficiently small fabrication defect rate, the likelihood of discarding any given device is vanishingly small. 

To deal with a level $j$ cluster of fabrication defects, we quarantine all of the qubits within the smallest square that contains the cluster. Hence, this square of inactive qubits has a linear size at most $Q^j$. We then use the intact unit cells at the boundary of the square to measure the super-stabilizers around the perimeter of the cell; these give rise to the shells introduced in Section~\ref{sec:UCnCD}. Due to the additional space required to construct the gauge operators, the super-stabilizers have a linear size at most $Q^j+2$ (c.f. Fig.~\ref{Fig:UnitCell}).

As we have proposed, we alternate between the two types of super-stabilizer with a frequency that is inversely proportional to the size of the fabrication defect cluster. For simplicity, we choose to alternate the type of super-stabilizer every $Q^j$ rounds for a cluster of level $j$. This results in shells of level $j$ with a spatial side length at most $Q^j+2$ (given as the linear size of a super-stabilizer) and the temporal side length of $Q^j$ (See Fig.~\ref{fig:TerminationP}). We define the diagonal width of the shell to be the Manhattan distance between two corners of the shell with the largest separation. Therefore, the diagonal width of a level $j$ shell is at most $2(Q^j +2)+Q^j = 3Q^j + 4$.

In the rest of the subsection, we shall use these results to prove the following bounds. First, we bound the number of error string termination points for individual shells, second, we bound the distance between the same type shells. Both temporal and spatial directions need to be considered in each case. These results allow us to upper bound the number of error configurations and lower bound the weight of error strings that span shells across the space-time lattice.

\begin{lemma}
\label{lemma:termination}
For a smooth or rough shell of level $j$, there are at most $40Q^{2j}$ points where respective error strings may terminate on that shell.
\end{lemma}

\begin{proof}
It is enough to examine the largest size smooth shell of level $j$ as any rough shell of the same level is strictly smaller by construction and, hence, has less termination points. The spatial and temporal length of a smooth shell of level $j$ is bounded by $Q^j + 2$ and $Q^j$ respectively, therefore the four sides of the shell has at most $4(Q^j +2)(Q^j +1) = 4Q^{2j} + 12Q^j + 8$ points where bit-flip errors may terminate without signature. Finally, we account for temporal errors due to the erroneous initial and final measurements of the other gauge. There are at most $4(Q^j + 1)$ such termination points at both - the top and bottom of the shell. See Fig.~\ref{fig:TerminationP} for a visual guide. Adding all of the contributions together we upper bound the number of termination points of any shell of level $j$ as $4Q^{2j} + 20Q^j + 16$. For simplicity of later calculations we further upper bound this number by $40Q^{2j}$, where we have used that $j\geq 0$ and $Q$ is a strictly positive integer. 
\end{proof}

In addition, our choice of gauge alternation frequency allows us to prove the following Lemma about the shell separation on a space-time lattice:
\begin{lemma}
\label{lemma:separation}
Let $Q\geq 9$, then any shell $S_j$ of level $j$ is separated from any other shell of level $k$ of the same type, $S_k$, as 
\begin{equation}
\label{eq:distS}
    D(S_j, S_k) \geq Q^{\mathrm{min}(k,j)},
\end{equation}
where distance $D$ is defined on the space-time lattice of unit cells and corresponds to the infinity norm.
\end{lemma}
\begin{proof}
To prove this Lemma we bound the spatial and temporal separations of the shells independently. First, we bound the spatial separation of shells that belong to different clusters. Then, we bound the temporal separation of shells that belong to the same cluster. Combining these results we arrive at our result. 

We start by considering the spatial separation $D_s$ of shells belonging to different clusters. Recall that a cluster of level $j$ with linear size at most $Q^j$ is separated from any other cluster of level $k \geq j$ by spatial distance at least $Q^{j+1} / 3$. Since the shells are centered around clusters and have a spatial size at most $Q^j+2$, any shell $S_j$ is spatially separated from any other shell $S_k$ with $k\geq j$ and belonging to a different cluster by $D_s(S_j, S_k) \geq Q^{j+1} / 3 - 2$. Using that $Q \geq 9$ and $j\geq 0$ we find $D_s(S_j, S_k) \geq 9Q^j / 3 - 2 = 3Q^j - 2 \geq Q^j$. The same calculation applies for the case when $j \geq k$ but with $j$ interchanged with $k$. Therefore, $D_s(S_j, S_k) \geq Q^{\mathrm{min}(k, j)}$ irrespective of the relative sizes of $k$ and $j$.

Now, consider the temporal separation $D_t$ of shells belonging to the same cluster. Recall, that we alternate between smooth and rough type shells every $Q^j$ time steps around a cluster of level $j$. Therefore, any two same type shells $S_j$ and $S'_j$ around the same cluster (hence the shell level $j$ for both) have a temporal separation $D_t(S_j, S'_j) \geq Q^j$.

Given that we use the infinity norm, the total separation $D(S_j, S_k)$ is the maximum between spatial and temporal separations. For the shells belonging to different clusters we lower bound $D(S_j, S_k) \geq D_s(S_j, S_k)$ while for the shells belonging to the same cluster we use $D(S_j, S_{k=j}) = D_t(S_j, S'_{j})$.
By combining these results we have proven the Lemma.

\end{proof}

While in general this is a loose lower bound, it allows us to consider separations across space and time in a homogeneous way. 

With this characterisation of the defect error model, it remains to determine if the functional qubits of the qubit array can successfully deal with generic random errors, namely bit-flip, phase-flip, and stabilizer readout errors. We assume that the intact cells are subject to a standard phenomenological noise model where their qubits suffer an independent and identically bit flip or phase flip error with probability $\epsilon$, and a stabilizer or a gauge operator returns an incorrect measurement outcome with probability $q = \epsilon$. We concentrate on a single type of error, namely, bit-flip errors and measurement errors on Pauli-Z type measurements, as we can deal with Pauli-X and Pauli-Z errors separately with the surface code~\cite{Dennis2002}. 
Likewise, we only consider shells with smooth boundaries, as these shells detect error strings of Pauli-X operators.
An equivalent proof will hold for the conjugate type of error, with only small changes to the microscopic details of the argument.
We expect that the generalisation of our results to a circuit noise model is straightforward with only small changes to the constant factors in our proof, see e.g.~\cite{Fowler2012,Tomita2014}.

\subsection{A threshold theorem for phenomenological noise}
\label{subsec:QT}

In this Subsection we prove that the construction of shells on a planar surface code architecture gives rise to a scalable quantum memory under the phenomenological noise model as long as the qubit array is operable.

\begin{theorem}
Suppose a space-time lattice that is generated from an operable qubit array of linear size $d$ and code deformations as presented above.
Then, there exists a phenomenological error rate threshold $\epsilon_0$ such that for independent and random errors with rate $\epsilon< \epsilon_0$ the probability of the logical failure is at most
\begin{equation}
\overline{P} = \exp(-\Omega(d^\eta))
\end{equation}
for some constant $\eta > 0$.
\end{theorem}

The logical failure rate is expressed as follows
\begin{equation}
    \overline{P} = \sum_{E \in\mathcal{F}} p(E), 
\end{equation}
where $\mathcal{F}$ denotes the set of errors that lead to a logical failure, and $p(E)$ denotes the probability that the error $E$ is drawn from the error model. We prove the theorem by showing that $\overline{P}$ rapidly vanishes as the code distance diverges for a suitably low but constant error rate.

To upper bound the logical failure probability we must characterise the set of errors $\mathcal{F}$. For our proof we adopt the tools from the threshold theorem of Dennis {\it et al.}~\cite{Dennis2002} and therefore we will presume the availability of a decoder that can find a least-weight correction. 
This correction can be obtained efficiently using the minimum-weight perfect-matching algorithm~\cite{Stace2009}.

The decoder fails when its correction error chain $E_{min}$, together with $E$ traverse some non-trivial path $\bm{\ell}$ such that  $E+E_{min}$ is a logical operator, and thereby has some minimal length $\ell \ge L$. It is argued in~\cite{Dennis2002} that the weight of the error $|E|$ must be at least $\ell / 2$, where we use $\ell$ to denote the weight of a path $\bm{\ell}$. Otherwise the decoder will either successfully correct the error, as the error itself will be a lower-weight correction, or create a new path for which the same condition applies. With this characterisation of error configurations that lead to a logical failure, we can regard errors $E \in \mathcal{F}$ in terms of error configurations lying along non-trivial paths $\bm{\ell}$. This enables us to upper bound $\overline{P}$ as follows
\begin{equation}
     \overline{P} \le \sum_{\ell \geq L} N(\ell) \Pi_{\ell}, \label{Eqn:LogicalFailureRate}
\end{equation}
where we sum over all non-trivial paths of length at least $L$ and $N(\ell)$ denotes the number of error configurations that lead the decoder to give rise to a non-trivial path of length $\ell$.
The associated weighting $ \Pi_{\ell}$ is bounded by the worst case probability that an error configuration of weight at least $\ell / 2$ occurs
\begin{equation}
\label{eq:weight}
   \Pi_{\ell}\leq \epsilon ^{\ell / 2}.
\end{equation}

In the original work of Dennis {\it et al.}~\cite{Dennis2002}, the authors upper bound $N(\ell)$ by recognising that the number of non-trivial paths $\bm{\ell}$ is upper bounded by the number of random walks of the same length, $\sim V \nu^\ell$, where $\nu$ is the valency of the lattice and $V$ is the number of starting points for the walk. The number $V$ is upper bounded by the volume of the space-time lattice which is polynomial in the code distance $d$. 
They obtain $N(\ell) \lesssim 2^\ell V \nu^{\ell}$  where $2^\ell$ upper bounds the number of ways to distribute $\ell/2$ or more errors along a path of length $\ell$. In the ideal case the minimum length of any non-trivial path is determined by the code distance $\ell \ge L = d$.

We generalise this expression to find an upper bound on the logical failure rate for surface-code error correction realised using a qubit array with fabrication defects. In this case the number of random walks changes when we consider a space-time lattice with shells. 
Additionally, shells compromise the distance of the code, 
hence, we must also obtain a lower bound on the length $L$ of non-trivial paths over the lattice.
In the remainder of this section we prove our theorem by upper bounding $N(\ell)$ and lower bounding $L$ for the case of a space-time lattice with shells. We note that we consider paths travelling in the spatial and temporal directions because, in the generality of fault-tolerant logical operations, logical errors may occur due to paths traversing the temporal direction. We note that formulation of our technical tools are agnostic of these details such that the space and time-like directions of the space-time lattice are isotropic.

Let us first formalise the definitions and introduce the rest of the terminology required. We characterise errors in $\mathcal{F}$ in terms of walks of a given length that contain some requisite number of errors in their support. We denote the length of a walk by $\ell$ where the walk can make progress with steps along the spatial or temporal directions of the space-time lattice. A spatial step of the walk is to move from one cell to any other cell that shares an edge with it (a qubit error) at some time $t$, while a temporal step of the walk is to stay on the same cell from time $t$ until time $t+1$ (a measurement error, either a stabilizer or a gauge operator). We say that a walk is \textit{non-trivial} if it starts at one boundary of the space-time lattice and ends at the respective other boundary. These boundaries should not be confused with the boundaries of shells.

A walk \textit{encounters} a shell if the error string  along  the  walk  terminates  at the  shell. For example, a walk of bit-flip errors may encounter a shell with smooth boundaries. After encountering a shell, the walk can continue from any other point of the same shell.
As in Ref.~\cite{Dennis2002} we only explicitly consider \textit{self-avoiding} walks for which every site, including shells, is visited at most once. This is justified because walks that cross any given site, or shell, more than once has an equivalent action on the code space as an error configuration that contains a self-avoiding walk. Moreover, walks that are not self avoiding are implicitly accounted for by our weighting function $\Pi_{\ell}$.

We are now ready to prove the following lemma, which is helpful to obtain bounds on $N(\ell)$ and $L$.

\begin{lemma}
\label{lemma:enc}
Let $Q \geq 9$ and assume an operable qubit array. Then on the space-time lattice any self-avoiding non-trivial walk $\bm{\ell}$ may encounter at most 
\begin{equation} 
n_j \le  2\times \frac{15^{j}\ell}{Q^j} \end{equation} 
shells with smooth boundaries of level $j$. Moreover, as a consequence, the length of an augmented walk, i.e, a walk that contains shells of all sizes up to $k$, stretches over a length
\begin{equation}
\ell'_k \le 15^{k+1}\ell. \label{Eqn:Augmented}
\end{equation}

\end{lemma}

\begin{figure}[tbp]
    \centering
    \includegraphics[width =1\linewidth]{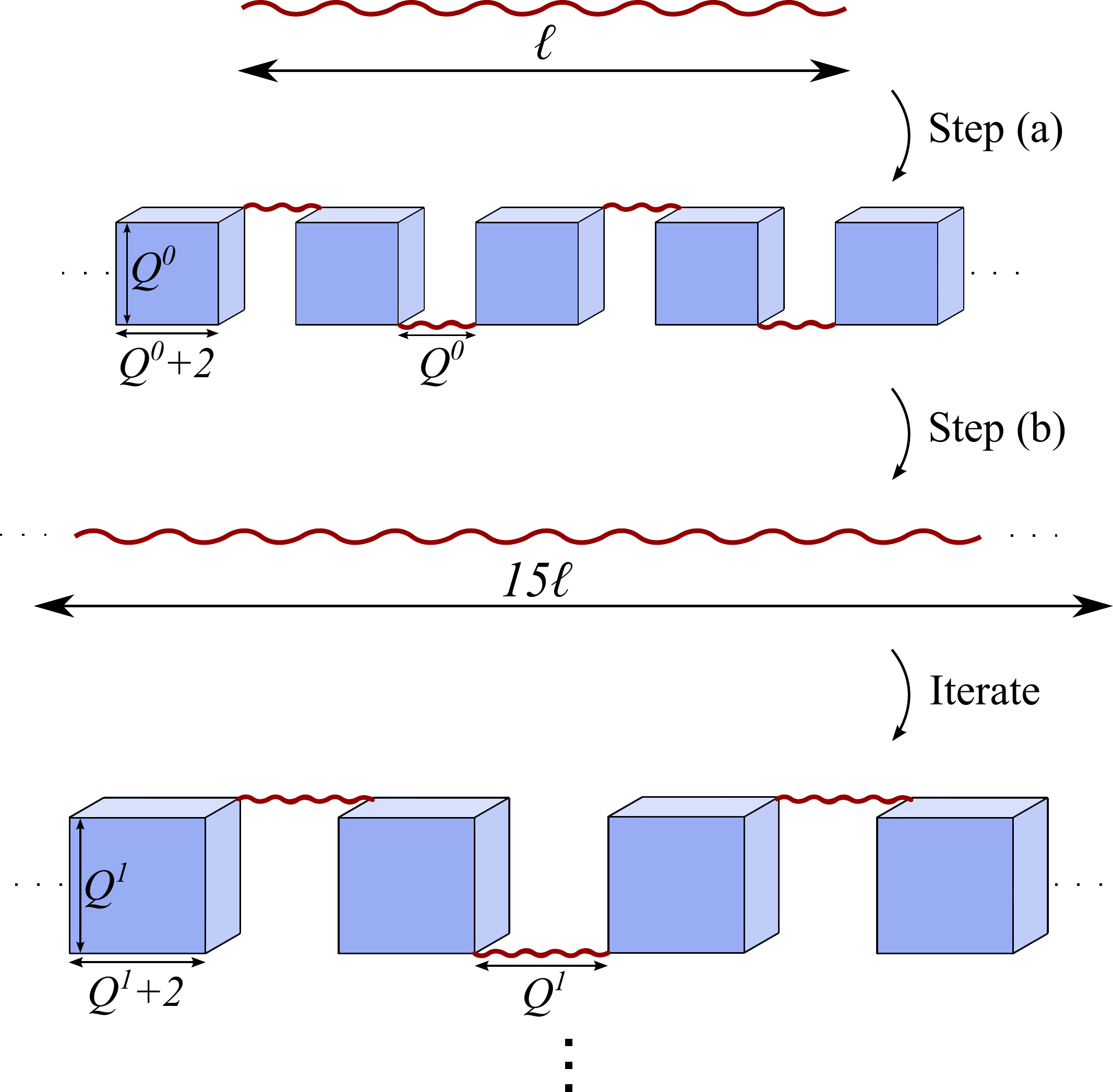}
    \caption{Diagram showing the first steps of the proof of Lemma~\ref{lemma:enc}. We start with a self-avoiding walk of length $\ell$ at the top. Step (a) inserts the maximum number of level $j=0$ shells that a walk can encounter. Step (b) then replaces this augmented walk with a shell-free walk of the same length. We repeat the same steps until the largest level $m$ shells have been accounted for.}
    \label{fig:Lemma}
\end{figure}

\begin{proof}

Consider a long walk in our space-time lattice. It may encounter multiple shells, and the diagonal width of these shells will contribute to the effective walk length. 

We wish to bound the number of shells that the walk may encounter. To do so, we start from a shell-free walk, and insert shells systematically at the highest possible density. We then make a considerable simplification: We simply regard our augmented walk as a shell-free walks of the new, greater length. We then go on to repeat the same shell-insertion now with shells one level greater in size. This will lead to an over-counting of the number of shells encountered, but it will permit us to reach an analytic expression. The approach is illustrated in Fig.~\ref{fig:Lemma}.

We prove the Lemma by induction. Given that $Q\geq 9$, Lemma~\ref{lemma:separation} shows that at most $n_0 = \ell/ Q^0 +1 $ shells of level $j=0$ can be encountered by a walk since any two shells are at least $Q^0$ separated. The $+1$ term arises from considering shells at both ends of the walk. We may further simplify the analysis by upper bounding $n_0 \leq 2\ell/Q^0 $, where we have used the fact that any non-trivial walk has a length $\ell \geq Q^0$ on an operable qubit array. Moreover, the length of the walk $\ell$ together with the diagonal width of all level $0$ shells, $3Q^0 +4$, have a combined effective length at most $\ell_0' = n_0(3Q^0+4)+\ell \leq n_07Q^0 + \ell \leq 14\ell + \ell =  15\ell$. This concludes the base case for $j = 0$.

Next, we assume the Lemma holds for $j=k$ to show that it holds for $j=k + 1 \leq m$. Consider a new shell-free walk with length $\ell_k'$ instead, this allows us to over-count the number of $k+1$ shells encountered by our original walk of length $\ell$. There can be at most $n_{k+1} = \ell_k'/Q^{k+1} +1 \leq 2 \times 15^{k+1}\ell/Q^{k+1}$ shells of level $k + 1$ that can be encountered by such a walk since any two shells of this level are at least $Q^{k+1}$ separated. Here again we have simplified analysis by converting the $+1$ term to a multiplicative factor of $2$.
We justify this simplification due to the fact that every non-trivial walk on the space-time lattice will be long enough such that their respective augmented walk must have length $\ell_{k}' \geq Q^{k+1}$ for $k \leq m -1$ given that the operability condition is satisfied. A more general proof can be found in Subsection~\ref{subsec:Hadamard}.
Now the walk $\ell_k'$, which already accounts for all shells up to level $k$, together with all level $k+1$ shells has a combined effective length at most $\ell_{k+1}' = n_{k+1}(3Q^{k+1}+4)+15^{k+1}\ell \leq n_{k+1}7Q^{k+1} + 15^{k+1}\ell  \leq  15^{k+2}\ell$.

We repeat the inductive steps until the largest shells of level $m$ have been accounted for.
\end{proof}

With Lemma~\ref{lemma:enc} in hand we can obtain a bound for Eqn.~\ref{Eqn:LogicalFailureRate}. Specifically, we can now find an upper bound for $N(\ell)$, the number of error configurations lying on a path of length $\ell$, and a lower bound for the length of the shortest non-trivial path, $L$.

We recall that the key technical step to finding an upper bound for $N(\ell) \leq  V \times W(\ell) \times C(\ell)$ is to count the set of self-avoiding non-trivial walks of length $\ell$, denoted $V \times W(\ell)$. This upper-bound is justified because the set of random walks of length $\ell$ must contain all of the non-trivial paths over the space-time lattice of the same length.

We multiply $V\times W(\ell)$ by the number of error configurations of at least $\ell / 2$ errors that can lie on a path of length $\ell$. We denote this number $C(\ell)$. It is straightforward to check that $C(\ell) < 2^\ell$. Moreover, the volume of the lattice is $V = d^3$ gives us an upper bound on the possible start locations of a walk. As such we proceed to obtain the number of random walks from a fixed starting point, $ W(\ell)$ where the walks may proceed through some number of shells that we have upper bounded using Lemma~\ref{lemma:enc}.

The first step of the walk may be made in $\nu =  6$ directions in the space-time lattice with $\nu$ the valency of the lattice, any subsequent step can then proceed in at most in $\nu - 1 = 5$ directions in order for the walk to be self-avoiding. We perform $\ell-1$ such steps. This bounds the number of walks with different step configurations as $(6/5)\times 5^{\ell}$.

To upper bound $W(\ell)$ we must account for the shells that the random walk may encounter. We have upper bounded the number of such encounters using Lemma~\ref{lemma:enc}. We have that
\begin{equation}
\label{eq:walkCount}
    W(\ell) \leq 6 \times 5^{\ell-1}\prod_{j=0}^m (40Q^{2j})^{2\times 15^{j}\ell/Q^j},
\end{equation}
where $40Q^{2j}$ is an upper bound of termination points of a level-$j$ shell (Lemma~\ref{lemma:termination}) that accounts for the number of distinct steps along which a walk can recommence after an encounter, and the exponent $n_j < 2\times 15^{j}\ell/Q^j$ is the upper bound on the number of times a walk of length $\ell$ will encounter distinct level-$j$ shells.
We take the product over shells of all levels that appear on an operable qubit array $j\le m$.

Let us rearrange Eqn.~\ref{eq:walkCount} to obtain a clearer expression for the logical failure rate. We first express the product in Eqn.~\ref{eq:walkCount} as a summation in the exponent. We also take each summation up to infinity to give an upper bound for its value. We find
\begin{align}
 \prod_{j=0}^m (40Q^{2j})^{2\times 15^{j}\ell/Q^j} \leq &
     40^{2 \ell\sum_{j=0}^\infty 15^{j}/Q^j}  \nonumber
     \\ &  \times Q^{4 \ell\sum_{j=0}^\infty j15^{j}/Q^j}.
\end{align}
Finally, using standard geometric series expressions with $Q > 15$ such that $x = 15 / Q < 1$, namely

\begin{equation}
\sum_{n=0}^\infty x^n = \frac{1}{(1-x)}, \quad \textrm{ and }\quad  \sum_{n=0}^\infty n x^n = \frac{x}{(1-x)^2},
\end{equation} 
we obtain
\begin{equation}
    W(\ell)   \leq   6 \times 5^{\ell-1}\times 
     40^{2\ell Q/(Q-15)} \times Q^{60\ell Q/(Q-15)^2}.
\end{equation}
We can therefore upper bound $N(\ell )$ as
\begin{equation} \label{Eqn:ConfigsUpperBound}
    N(\ell) < V W(\ell) C(\ell) \leq c d^3 \rho^\ell
\end{equation}
where $c = 6/5$ and $\rho = 2\times 5 \times 40^{2Q/(Q-15)}Q^{60Q/(Q-15)^2}$ are constants.

The lower bound $L$ is also readily obtained from Lemma~\ref{lemma:enc}. Let $d$ be the linear size of the space-time lattice, then any non-trivial augmented walk must have an effective length $\ell_m' \geq d$. 
Using Eqn.~\ref{Eqn:Augmented} we have that the maximum effective length that a walk of length $\ell$ can achieve is  $\ell_m' \le 15^{m+1}\ell$ when $m \leq \log_Q   (d/2 - 1)$ in the case of an operable qubit array.
Therefore, any non-trivial walk must be at least of length
\begin{equation}
    L = 15^{-(m+1)}\ell_m' = 15^{-(m+1)}d. \label{eq:Shortest}
\end{equation}

We can now combine the equations we have obtained to complete our proof. Substituting our expression for $N(\ell)$ given in Eqn.~\ref{Eqn:ConfigsUpperBound} into Eqn.~\ref{Eqn:LogicalFailureRate} together with Eqn.~\ref{eq:weight} enables us to bound the logical failure rate as follows
\begin{equation}
    \overline{P}  \leq c d^3  \sum_{\ell\geq L} \rho^\ell \epsilon^{\ell/2} = 
  k d^3 \left(\rho \epsilon^{1/2}\right)^{L}  = 
    k d^3 \exp(-\Theta(L)), \label{Eq:Intermediate}
    \end{equation}
    with $ k = c / (1 - \rho \epsilon^{1/2}) $ a positive constant term  for $\rho \epsilon^{1/2} < 1$. 
    
Substituting Eqn.~\ref{eq:Shortest} into Eqn.~\ref{Eq:Intermediate} reveals the exponential decay in the logical failure rate. We have that
    \begin{equation} \label{eq:finale}
       \overline{P}  \leq  
     k d^3 \exp(-\Theta(15^{-(m+1)}d)) \leq k d^3 \exp(-\Theta(d^\eta)),
\end{equation}
where $\eta = 1 - \log_Q(15)$.
By observation we obtain an exponential decay in logical failure rate provided $Q > 15$ and $\rho \epsilon^{1/2} < 1$. A lower-bound for the threshold for phenomenological noise is therefore $ \epsilon_0 \equiv 1/\rho^2 $, thereby completing the proof of our main theorem.

\section{Universal quantum computation}
\label{sec:FTQC}

Universal quantum computation can be performed with the surface code using code deformation~\cite{Dennis2002, Raussendorf2006, Bombin2009, Fowler2012, Horsman2012, Brown2017}. Indeed, the computational techniques we have developed in the previous section are readily used to demonstrate a vanishing logical failure rate for fault-tolerant logic gates using a surface code architecture with fabrication defects.

Let us show that universal fault-tolerant quantum computation is possible with a planar array of qubits with fabrication defects. For this proof we choose a minimal set of universal gates but remark that, in principle, our arguments can be adapted for an arbitrary choice of code deformations with the surface code~\cite{Brown2017, Litinski2019}. We take a Hadamard gate as presented in Ref.~\cite{Dennis2002} and we perform entangling operations with lattice surgery~\cite{Horsman2012}. Given these fault-tolerant gates, we can complete a universal gate set by distilling magic states~\cite{Bravyi05} and Pauli-Y eigenstates~\cite{Raussendorf2006, Fowler2012}. We conclude this section by showing that the logical error rate of magic-state injection is bounded on a defective lattice.

\subsection{Hadamard gate}
\label{subsec:Hadamard}

In Ref.~\cite{Dennis2002} the authors propose performing a Hadamard gate by rotating the corners of the planar code about its boundary~\cite{Brown2017} before performing a transversal Hadamard operation followed by a swap operation. See also Ref.~\cite{Fowler2012}. The latter unitary operation only maps between the two different types of error strings through a plane in the space-time lattice, so let us concentrate on the physical transformation to the lattice boundaries. We show this operation of the space-time lattice with shells in Fig.~\ref{fig:HadamardG}. 

The transformation of the surface code boundaries must be done in way that preserves the code distance. As one can see in Fig.~\ref{fig:HadamardG}(b), during the switch of the boundaries, we may have purely diagonal paths of errors that lead to a logical error compared to the static quantum memory case. For the threshold theorem to hold, we only need to reconsider the proof of Lemma~\ref{lemma:enc} by taking into account such diagonal paths.
Lemma~\ref{lemma:enc} already captures the diagonal extent of the shells, however, for it to hold, one needs to show that every non-trivial walk will be still long enough such that their respective augmented walk has length $\ell_{k}' \geq Q^{k+1}$ for $k \leq m-1$. We prove it formally.

\begin{proof}
Consider a single largest shell of level $m$ along a non-trivial walk. Then given an operable qubit array ($Q^m \leq d/5 -3$), the diagonal width of such a shell is at most $3Q^m + 4 \leq 3/5d - 5$, where $d$ is the distance of the code. Since $d$ is preserved during the transformation of boundaries, the respective augmented walk of a non-trivial walk with all shells of level $\leq m-1$ has to be at least $\ell_{m-1}' \geq d-(3/5d - 5) = 2/5d + 5 \geq d/5 -3 \geq Q^m$, where we have subtracted the maximum diagonal width of a single level $m$ shell. If there are more than one shell of level $m$ along the non-trivial walk, then by Lemma~\ref{lemma:separation}, their mutual separation is at least $Q^m$ and the result holds trivially.

Moreover, having this result hold for the biggest shells of level $m$ implies that it holds for all shells of level $\leq m$. We show it by induction, where we have already proven that $\ell_{j}' \geq Q^{j+1}$ for the base case $j = m-1$. 

Now, we show that if the result holds for the $j = k$ case then it also holds for $j = k-1$ case. Consider a single shell of level $k$ along a non-trivial walk. It has a diagonal width of $3Q^{k} +4$. Given that $\ell_{k}' \geq Q^{k+1}$ holds from the $j = k$ case, the respective augmented walk of a non-trivial walk with all shells of level $\leq k-1$ has to be at least $\ell_{k-1}' \geq Q^{k+1} - (3Q^{k} + 4) \geq Q^{k}$. Similar to the previous argument, if there are more than one shell of level $k$ along the non-trivial walk, then by Lemma~\ref{lemma:separation}, their separation to any shell of level $k$ or greater is at least $Q^{k}$ and the result holds trivially. 

We may repeat the inductive steps until we have proven that $\ell_{k}' \geq Q^{k+1}$ for any $0 \leq k\leq m - 1$.
\end{proof}

This proof ensures that Lemma~\ref{lemma:enc} holds. The rest of the proof of the threshold follows from the memory case.

\begin{figure}[tbp]
    \centering
    \includegraphics[width =1\linewidth]{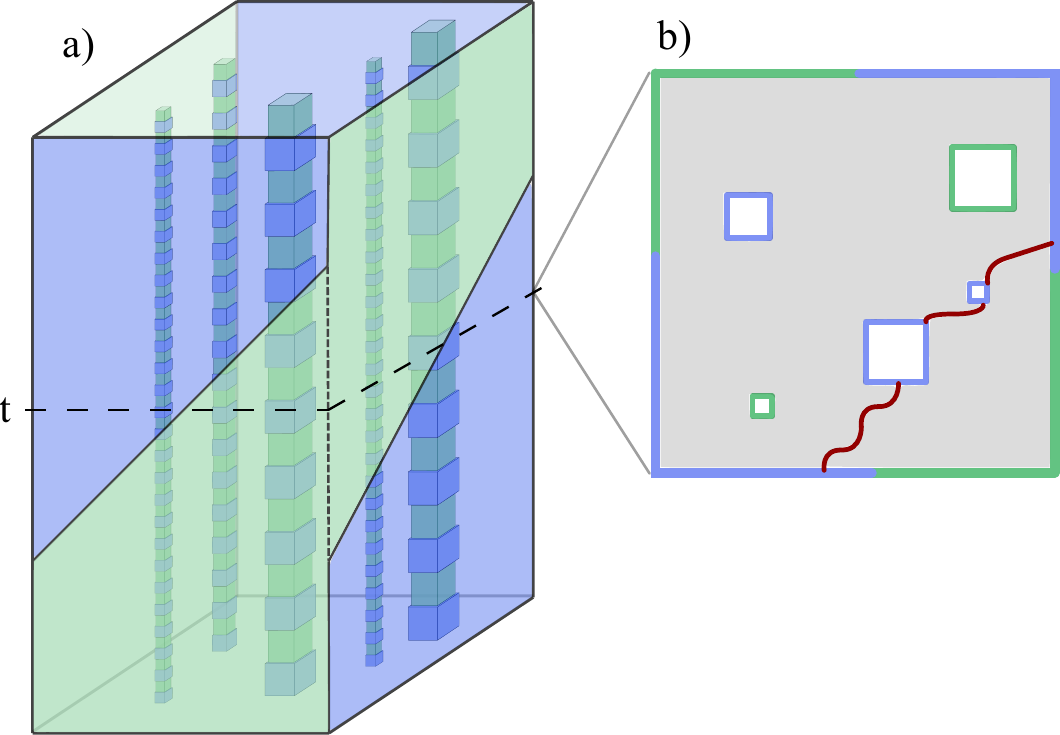}
    \caption{(a) Implementation of a Hadamard gate in the space-time picture. We show a possible path of errors (b) that lead to a logical failure during the switch of the boundaries at time $t$.}
    \label{fig:HadamardG}
\end{figure}

\subsection{Lattice surgery}

Lattice surgery~\cite{Horsman2012} was introduced to perform entangling operations between surface codes by parity measurements. In its simplest form, two adjacent surface code patches are merged using a routing space. The parity measurement is made using low-weight stabilizer measurements of the merged code, as the logical parity measurement is a member of the stabilizer group of the merged code. Once we have repeated the stabilizer measurements of the merged codes over a long time; $O(d)$ rounds of stabilizer measurements, we separate the two codes to complete the parity measurement by collapsing the routing space with a product state measurement.

One should check that we can obtain the value of the logical parity measurement if the routing space contains a fabrication defect.
Naturally, our shell construction can be used to reliably evaluate the value of the necessary super-stabilizers needed to complete the logical parity measurement near to the punctures.

It is readily checked that the logical failure rate of lattice-surgery operations is suppressed using a qubit array with fabrication defects.
Broadly speaking, assuming an operable qubit array, the logical failure rate is governed by two factors. The dominant factor is the exponential suppression of a non-trivial path being introduced to the space-time volume. This factor overwhelms the other factor of the polynomial volume of the space-time lattice. 

We restate this in a more formal sense.
Unlike the case of a quantum memory that we have already considered, logical errors during lattice surgery operations can also occur if string-like errors extend along the temporal direction of the routing space while the codes are merged. In Lemma~\ref{lemma:enc} we have already considered paths that traverse the temporal direction. As such, the only quantitative change from the memory case proof is the factor $V$ denoting the starting points of a walk. Considering that the volume of an efficient lattice surgery operation, which includes the volume of the two codes as well as that of the routing space, is also polynomial, $\sim d^3$, the number of logical string starting points is of order $V = O(d^3)$. The exact number depends on the size of the routing space. Therefore, we can rewrite the Eqn.~\ref{eq:finale} as
\begin{equation}
           \overline{P} \leq k O(d^3) \exp(-\Theta(d^\eta)),
\end{equation}
where the exponential suppression of the logical error probability with respect to the code distance is apparent.

We finally remark that, given a qubit array that is large enough to perform a lattice-surgery operation, it is exponentially likely with respect to the code distance that we will obtain an operable qubit array, assuming the fabrication error rate is sufficiently low. This is again due to the fact that the width of the array needs only to be of order $O(d)$.

\subsection{$S$ and $T$ gates}

\begin{figure*}[tbp]
    \centering
    \includegraphics[width =0.7\linewidth]{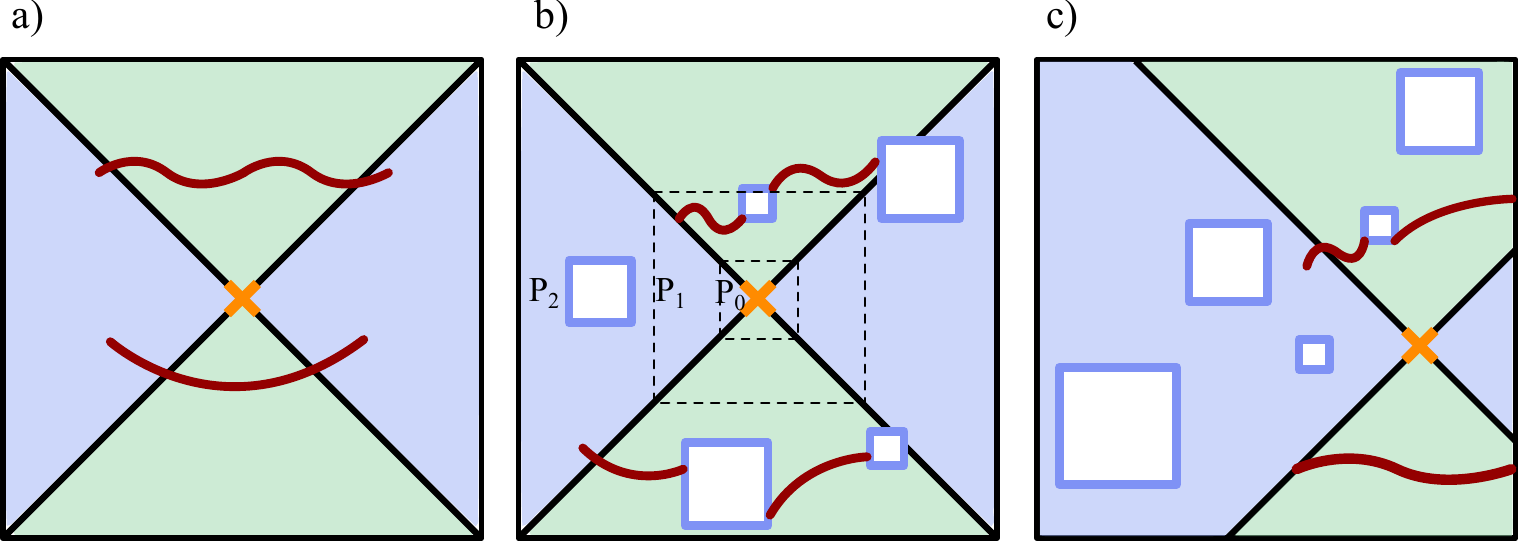}
    \caption{(a) State injection scheme, where the injection qubit is denoted with an orange cross and blue(green) qubits are initialised in the $|+\rangle$($|0\rangle$) basis. Error strings in space-time (red) may form from one boundary to the respective other boundary and cause a logical failure. (b) Fabrication defects increase the noise acting on the logical input state. Panels $P_j$ are defined in the main text whereby no level-$j$ shell is supported on panel $P_j$. The orange cross marks a \textit{good injection point}. (c) In general, a good injection point may be found off-centre.}
    \label{fig:MSI}
\end{figure*}

Given the aforementioned logical gates, we can complete a universal gate set with distillation circuits. We complete the Clifford gate set with a distillation circuit for an eigenstate of the Pauli-Y operator~\cite{Raussendorf2006, Fowler2012}. This is done using noisy state preparation, together with the Hadamard gates and lattice surgery operations that we have already discussed. Then, given the ability to perform all of the Clifford gates and inject noisy ancilla states, we can use magic state distillation to perform universal quantum computation~\cite{Bravyi05}.

It remains to discuss state injection~\cite{Li2015, Lodyga2015} on a two-dimensional array with the independent and identically distributed fabrication defect model we have assumed throughout this work. In the ideal case, we initialise noisy ancilla states by preparing the qubit array with some qubits in the $|+\rangle$ state and others in the $|0\rangle $ state, together with a single qubit prepared in the input state $|\psi \rangle$. We then begin measuring the stabilizers of the surface code to encode the input state. We show a suitable input configuration in Fig.~\ref{fig:MSI}(a), where blue(green) qubits are initialised in the $|+\rangle$($|0\rangle$) basis, respectively, and the central qubit is prepared in the state $|\psi \rangle$.

Provided we can argue that our shell construction does not limit our ability to prepare an input state with a bounded with a logical error rate, we can fault-tolerantly implement both of the proposed gates by distillation. Here we show that we can initialise a state with a finite logical error rate that does not diverge with the distance of the code below some threshold rate of phenomenological noise. Our argument assumes that we can find a good injection point, as we define below. Once we argue that we can inject a state with bounded probability using one such injection point, we will also argue that we can find a good injection point on a region of an operable qubit array with high probability.

It is possible to initialise an ancilla state given a \textit{good injection point}, i.e., a region of the lattice that is distant from fabrication defects. We regard any unit cell $v$ as a good injection point if
\begin{equation}
\label{eq:FTQC1}
   D(v, u_j) \geq 2Q^j + 2
\end{equation}
is satisfied, where $u_j$ describes a unit cell contained within a level-$j$ shell. This means that a good injection point is further away from larger clusters of fabrication defects compared to the smaller ones.

It is now our goal to show that the logical error rate of an input ancilla state is bounded given a good injection point. The  main technical step we use to obtain an expression for the logical failure rate is to count the number of ways in which a logical error can occur. Again, we find this value in terms of random walks. We upper bound the number of error configurations by counting the number of shells a walk can encounter in the vicinity of a good injection point. We also require a lower bound on the length of a walk from a given starting location to obtain an upper bound on the error rate for the input ancilla state.

To begin, assuming a good injection point, we can use the result from Lemma~\ref{lemma:enc} to upper bound the number of shells that a non-trivial walk may encounter. It is justified to use Lemma~\ref{lemma:enc} because any non-trivial self-avoiding walk that encounters a level $j$ shell has $\ell_{j-1}' \geq Q^j$ where $\ell_{j-1}'$ is the augmented walk excluding level $j$ shells or larger. This justification is depicted geometrically in Fig.~\ref{fig:MSI}(b) and follows from proof in Subsection~\ref{subsec:Hadamard}. By letting $Q\geq 9$, we can therefore say that any non-trivial self-avoiding walk of length $\ell$ from one boundary to the respective other boundary encounters at most $ 2\times \frac{15^{j}L}{Q^j} $ shells of level $j$.

We next count the number of starting points from which a random walk can begin. To simplify the following calculation we discretise the qubit array in disjoint panels around the injection point as shown in Fig.~\ref{fig:MSI}(b). The lattice is discretised such that panel $P_j$ has edges that are a distance of $2Q^j + 2$ away from the good injection point. This enables us to upper bound the number of starting points $s_j$ of any walk origination from within panel $P_j$ as $s_j \leq (4Q^j +4)^2$.

Given the arguments we have presented we can upper bound the number of error configurations that give rise to a non-trivial path of length $\ell$ originating from panel $P_j$ as
\begin{equation}
    N_j (\ell) \leq c (4Q^j +4)^2 \rho^\ell,
\end{equation}
where $c = 6/5$ and $\rho = 2\times 5 \times 40^{2Q/(Q-15)}Q^{60Q/(Q-15)^2}$ are constants as before.

Next we need to find a lower bound on the minimal walk length to obtain the logical error rate. 
Any non-trivial walk starting from within a panel $P_j$ must have an effective length $\ell' \geq 2Q^{j-1}$. Due to Eqn.~\ref{eq:FTQC1} no shell of level $j$ may be found within the panel $P_j$.
This implies that the minimum length of any non-trivial walk originating from panel $P_j$ must satisfy $L_j \geq   15^{-j}\ell' \geq 15^{-j}(2Q^{j-1})$. We note that this expression accounts for the case where the injection point is found close to the boundary of the surface code, see Fig.~\ref{fig:MSI}(c).

Taking the expressions we have obtained for $N_j(\ell)$ and $L_j$ for each panel $P_j$ and recalling Eqns.~\ref{Eqn:LogicalFailureRate} and~\ref{eq:weight} we bound the logical failure rate
\begin{multline}
\label{eq:expS}
    \overline{P}_\textrm{in} \leq \sum_{j=0}^{n} \sum_{\ell \geq L_j} N_j(\ell) \epsilon^{\ell / 2} \\
    \leq c \sum_{j=0}^n s_j  \sum_{\ell \geq L_j} \rho^\ell \epsilon^{\ell / 2} = 
     k \sum_{j=0}^n s_j (\rho \epsilon^{1/2})^{L_j} \leq O(1)
\end{multline}
as long as $\rho \epsilon^{1/2} < 1$ which is satisfied for $\epsilon  < 1/\rho^2 \equiv \epsilon_0 $, where $k = c / (1-\rho \epsilon^{1/2})$ and $n\geq m+1$ denotes the total number of panels. Consequently, there must exist a phenomenological error rate independent of the system size below which $\overline{P}_\textrm{in}$ is low enough for the encoded qubit to be used in the distillation protocol.

Finally, we show that a good injection point can be found with exponentially high probability for a fabrication defect rate below some threshold. We argue that this is the case as long as $Q\geq 33$ and the biggest cluster of level $m$ has size $Q^m \leq d/5-3$. As can be seen in Ref.~\cite{Bravyi2011} this condition is satisfied with exponentially large probability with respect to the qubit array size below the fabrication defect rate $f_0 = (3Q)^{-4}$.

From the properties of cluster decomposition (c.f. Section~\ref{subsec:ErrorModel}) any cluster of level $j$ is separated from all other clusters of level $k\geq j$ by distance $Q^{j+1}/3$. Hence, for $Q\geq 33$ any two shells of the same level $j\geq 1$ are separated by distance at least $3(2Q^j+2)$. This allows us to find a region of linear size $2Q^j+2$ between these shells such that the cells within this region are separated by distance at least $2Q^j+2$ from all shells of level $j$. Now, this region may contain a shell of level $j-1$. We can use the same argument to find an even smaller region in which all cells are separated from all level $j-1$ shells by distance at least $2Q^{j-1}+2$. We can continue to iterate to smaller and smaller regions until level $0$. The region found at level $0$ corresponds to the set of good injection points. Note that we require $Q^m \leq d/5-3$ to ensure that for all $m$ the region of good injection points is found within the boundaries of the surface code.

We have found that given a good injection point we may encode a qubit in an arbitrary state with a finite rate of logical failure. Furthermore, we have shown that such a good injection point can be found with exponential probability with respect to the system size below some fabrication defect rate. Therefore, together with fault-tolerant implementation of lattice surgery and hadamard gates we find that we may perform state distillation protocols for both $S$ and $T$ logical gates fault-tolerantly.

\section{Cosmic rays and other time-correlated errors}
\label{sec:CR}

There are a number of events that can lead to errors that compromise a two-dimensional qubit architecture over a long time. Our technical results assumed we can determine the the locations of permanent time-correlated offline before we perform a quantum computation. However, there are time-correlated errors that we might expect our architecture to experience while our computation is online. Let us discuss here how we begin measuring a shell once a time-correlated error is detected, and how we might detect a time correlated error.

One example are the errors introduced by cosmic rays, where the qubit array absorbs a large amount of energy that significantly increases the error rate of the qubits in a region. We can expect the qubits to return to their operational state as the energy dissipates. However, this may take a long time, and it may be worthwhile isolating the qubits from the error-correcting system using our shell protocol while the energy dissipates. Leaving qubits that are above the threshold error rate within the system will be particularly problematic during computations, as it will be very difficult to learn the value of logical parity measurements reliably in operations such as, for instance, lattice surgery. Let us discuss how we might generalise our protocol to deal with more general time correlated errors. We show the idea in Fig.~\ref{fig:CosmicRays}.

Upon the discovery of a time-correlated error that occurs at time $t_1$ in Fig.~\ref{fig:CosmicRays}, we must begin the measurement of super stabilizers as discussed in Sec.~\ref{sec:UCnCD} to isolate the inoperable parts of the qubit array from the code. We begin by performing the single-qubit Pauli-X measurements at time $t_2$ to infer the value of $\mathcal{A}_P$ where $P$ is a puncture that completely encloses the time-correlated error. For now we assume we can rapidly determine a time-correlated error has occured within time $\Delta = t_2 - t_1$. Thereon we continue to perform our shell protocol as normal where we alternate between measuring the stabilizers corresponding to a rough and a smooth boundary about $P$ where we alternate with a frequency that is proportional to the radius of $P$. 

\begin{figure}[tbp]
    \centering
    \includegraphics[width =0.7\linewidth]{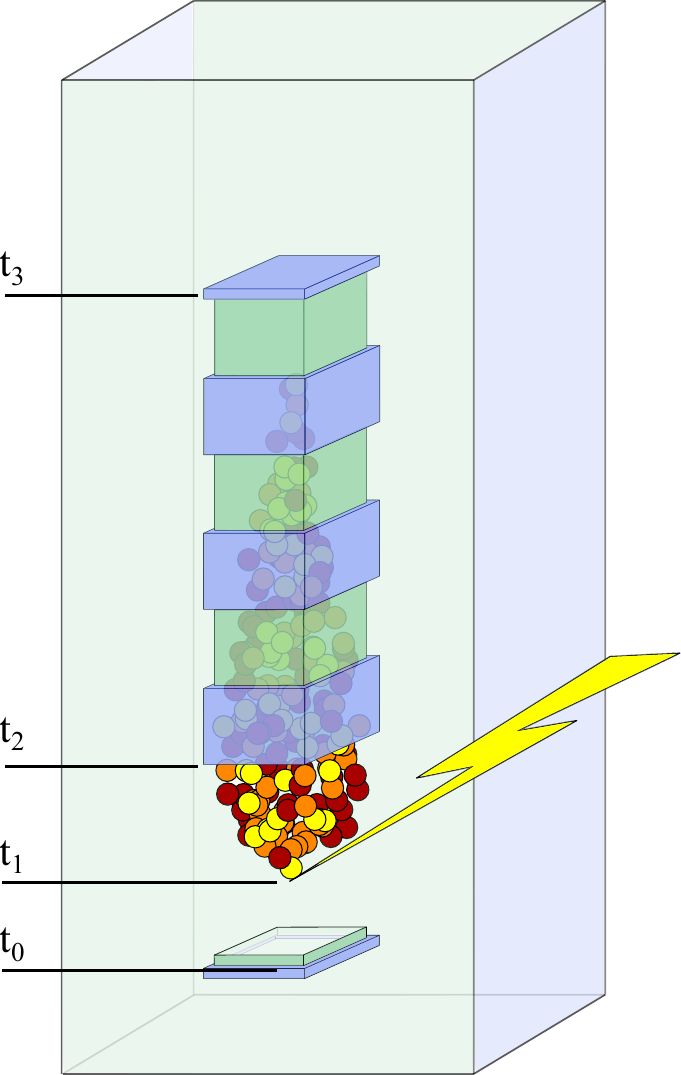}
    \caption{A cosmic ray event. A cosmic ray hits our qubit array at time $t_1$. We use the time between $t_1$ and $t_2$ to detect the event before isolating the affected qubits with shells. Stabilizer measurements before the cosmic ray hits (e.g. at $t_0$) can be used to infer the base value of the shells. We wait until the energy due to the cosmic ray dissipates before initialising all affected qubits back into the error-correcting system at time $t_3$.}
    \label{fig:CosmicRays}
\end{figure}

We can also measure shells that identify the parity of error in a space-time region that encloses the moment the time-correlated error began, $t_1$. Assuming the code was initialised much earlier on, we obtained the value of both super stabilizers reliably, that is, up to measurement errors that we can identify, at some time $t_0 < t_1$ before the moment the time-correlated error began. Indeed, these super-stabilizer operators are members of the stabilizer group before the time-correlated error occurred, and we can learn their values by taking the product of the star and plaquette operators before the time-correlated error began, see Fig.~\ref{fig:CosmicRays}. We can compare the values of these stabilizers to the first measurements of the super-stabilizers in the shell protocol to identify strings that enter into the region that contains the genesis of the time correlated error to obtain the values of the first shells associated to this puncture that is projected in the spacetime lattice.

We might also consider the discovery of a time-correlated error that occurs very soon after the surface code is initialised where there may not have been enough time to learn the value of the super-stabilizer. In which case, as discussed in Sec.~\ref{sec:UCnCD}, we can begin the shell protocol around the new time-correlated error to initialise the super stabilizers in a known eigenstate.

It may be that we can repair the qubits that experience a time-correlated error while the system is online. Likewise, qubits that have been impacted by, say, a cosmic ray, will become functional again once the energy from the radiation dissipates. In which case we can reinitialise the qubits in the surface code by measuring the stabilizers of the unpunctured code once again in the region $P$. This occurs at time $t_3$ in Fig.~\ref{fig:CosmicRays}. These measurements give a final reading of the super stabilizers to complete the measurement of the last shells that are measured about a given puncture. We simply compare these super-stabilizer readings to the last measurements of $\mathcal{A}_P$ and $\mathcal{B}_P$ that were inferred by the shell protocol. Irreparable qubits can remain isolated until the end of a computation.

We have thus far explained how we can prepare shells around a time-correlated error that appear appears while the qubit array is online running a computation. However, it remains to explain how we can identify the occurrence of time correlated errors. We argue that we can identify a time correlated error by monitoring the frequency that stabilizer detection events appear in time. 
To be clear, we are interested in a practical setting where error rates are, say, at least a factor of ten below the threshold error rate of the code and some time-correlated error occurs that increases the physical error rate above threshold. In which case, to leading order, the occurrence of stabilizer events will also increase by a factor of ten near to the initial event that caused the time-correlated error. We expect to be able to notice this change on a region of the lattice within a short time of its occurrence and isolate the qubits with the shell protocol very quickly. Experimental results~\cite{McEwen2021} have shown that cosmic rays can be identified very rapidly, i.e., on the timescale of a single-round of stabilizer measurements. It is therefore reasonable to assume that $\Delta = t_2 - t_1$ can be very small.

More generally, other work has already considered identifying drifts in noise parameters by measuring the frequency of the occurrence of stabilizer detection event to improve the performance of quantum error-correcting systems by updating the prior information given to decoding algorithms~\cite{Huo2017,Spitz2018, Nickerson2019analysingcorrelated}, or experimental control parameters~\cite{Kelly2016}. It may even be advantageous to use the shell protocol to isolate regions of the lattice where the error rate increases dramatically, even if the error rate of the respective qubits remains below threshold.

The spatial locations of time-correlated errors may drift over time. This more general case was considered in Ref.~\cite{Bombin2016}. We can monitor the statistics of the occurrence of stabilizer events near to punctures to detect drifts. Given that we can identify such a drift by monitoring the statistics of the occurrence of detection events near to a time correlated error, we can change the size and shape of our shells over time to isolate a mobile time-correlated error from the qubit array.

One other type of time-correlated errors are those where a stabilizer measurement device becomes unresponsive and always reports the no-event outcome, independent of the state of the qubit array. This was called a `silent-stabilizer error' in Ref.~\cite{Waintal19}. We can propose a straight forward method to identify such errors and thereafter treat them as time-correlated errors as discussed above. Indeed, we can continually test the responsiveness of our stabilizer readout circuits by manually applying Pauli rotations to our qubits in known locations such that, assuming no other errors occur, the state of a stabilizer degree of freedom is changed with probability one.

Let us elaborate on this proposal such that we can also identify the standard physical errors the code is designed to detect. One can find a Pauli operator that rotates \emph{every} stabilizer degree of freedom in between each round of stabilizer measurements. Now, we continue to measure stabilizers as normal with this additional step in the periodic readout circuit. However, instead of looking for detection events by looking for odd parity between two time-adjacent stabilizer measurements, we look for even parity between two measurements to detect an event. This circuit will identify silent stabilizer errors very quickly, as they will produce events at every single time step. We can therefore quickly identify the unresponsive stabilizer and isolate it from the system as a typical fabrication defect.

\section{Discussion}
\label{sec:Disc}

We have shown that we can perform an arbitrarily large quantum computation with a vanishing logical failure rate using a planar array of noisy qubits that has a finite density of fabrication defects. We obtained our result constructively, by proposing a syndrome readout protocol for two-dimensional architectures that are now under experimental development. This represents a substantial improvement upon existing proposals to deal with time-correlated errors that make use of single-shot error correction, as known single-shot codes require a three-dimensional architectures that may be more difficult to manufacture.

To employ our protocol it is necessary to establish the locations of fabrication defects or other anomalies such as cosmic ray events. This can be achieved either offline, or during the execution of computations by studying the frequency with which the qubit array produces detection events. In contrast, single-shot codes remain functional even if they operate unaware of information about the locations of time-correlated errors~\cite{Bombin2016}. It will be interesting to find the best ways of inferring the locations of time correlated errors in two-dimensional systems to improve their performance~\cite{Huo2017,Spitz2018, Nickerson2019analysingcorrelated, Flammia20, Harper20, wagner2021pauli}, as well as to compare their performance in terms of resource overhead to that of three-dimensional single-shot protocols~\cite{Bombin2015, kubica2021singleshot} for dealing with time-correlated errors.

Further afield, it will be interesting to consider the challenges that a finite rate of fabrication defects presents for the task of algorithm compilation~\cite{Gidney2019efficientmagicstate, Litinski2019}. We typically assume that all of the logical qubits and fault-tolerant gates involved in the compiled form of an algorithm will perform uniformly well. However, in practice, local regions of the qubit array that happen to have a larger density of fabrication defects will lead to lower-grade logical qubits for a given number of physical qubits. We would therefore wish our compiler to be {\it hardware-aware} and capable of adapting to such inhomogeneities. It might, for example, allocate a larger area of the qubit array to produce better logical qubits with a higher code distance in lower-grade regions. The discrepancy in code distance would then have consequences for the implementation of logical gates, potently giving rise to inefficiencies due to the fact that logical qubits that occupy differently sized or shaped regions of the qubit array may not tessellate. Alternatively, we might consider allocating the lower-grade regions of the array to heralded operations that can be performed offline. This will mean that the subroutines that are performed by more fault-prone logical qubits can be isolated from the larger calculation until they have been thoroughly tested for errors. Indeed, the compilation of topological circuits with an inhomogeneous qubit array remains an unexplored area of study.

We demonstrated a threshold using a protocol with a number of conservative simplifications to find an analytical expression for the logical failure rate. However, we can expect variations of our protocol that forego the simplifications or adapted variations of previously suggested methods~\cite{Auger2017} to perform better in practice when we take the system architecture in consideration. It will be valuable to compare them numerically to optimise the performance of a qubit array. Indeed, subsequent to this paper appearing online, two numerical studies based on our shell protocol have appeared which report highly encouraging performance with respect to realistic defect rates~\cite{siegel2022adaptive,lin2023empirical}.

A further route for future research would be to modify our proposal for other topological codes. Our protocol adapts readily to other variations of the surface code with an equivalent boundary theory~\cite{bravyi2013subsystem, Chamberland2020, BonillaAtaides2021, Higgott2021}, but we may also find that generalisations of our protocol using general topological codes~\cite{Kitaev03, Levin05, Bombin06, Hastings2021dynamically} with richer boundary theories~\cite{Beigi11,Kitaev12, Levin13, Barkeshli13, Kesselring18, vuillot2021planar} can outperform our surface code protocol. This future work will determine how best to nullify the harmful effects of fabrication defects in modern qubit architectures.

\begin{acknowledgements}
We thank R. Babbush, H. Bomb\'{i}n, N. Bronn, J.~Chavez-Garcia, A. Doherty, S. Flammia, A. Fowler, C. Gidney, C. Jones, M. Kjaergaard, J. Martinis, M. Newman, B. Vlastakis and A. Zalcman for conversations about fabrication defects and cosmic rays.
BJB is supported by the European Union’s Horizon 2020 research and innovation programme under the Marie Skłodowska-Curie grant agreement No. 897158, and also received support from the Australian Research Council via the Centre of Excellence in Engineered Quantum Systems (EQUS) project number CE170100009. SCB acknowledges support from the EPSRC QCS Hub EP/T001062/1, from U.S. Army Research Office Grant No. W911NF-16-1-0070 (LOGIQ), and from EU H2020-FETFLAG-03-2018 under the grant agreement No 820495 (AQTION). This research was funded in part by the UKRI grant number EP/R513295/1.
\end{acknowledgements}

\bibliography{export2.bib}

\end{document}